\documentclass[english]{article}
\usepackage[T1]{fontenc}
\usepackage[latin9]{inputenc}
\usepackage{color}
\usepackage{amsthm}
\usepackage{amsmath}
\usepackage{amssymb}
\usepackage{graphicx,threeparttable}
\usepackage{footnote,epstopdf}
\makesavenoteenv{tabular}
\makesavenoteenv{table}
\makeatletter

\newcommand{\lyxmathsym}[1]{\ifmmode\begingroup\def\b@ld{bold}
  \text{\ifx\math@version\b@ld\bfseries\fi#1}\endgroup\else#1\fi}

\theoremstyle{plain}
\newtheorem{thm}{Theorem}
\theoremstyle{plain}
\newtheorem{lem}[thm]{Lemma}

\@ifundefined{definecolor}
 {\usepackage{color}}{}
\@ifundefined{definecolor}{\usepackage{color}}{}
\usepackage{subfigure}

\global\long\def\rt{\operatorname{shift}}
\begin{document}

\title{Skolem Sequence Based Self-adaptive Broadcast Protocol in Cognitive Radio Networks}

\author{ 
Lin Chen$^{1,2}$, Zhiping Xiao$^2$, Kaigui Bian$^2$, \\Shuyu Shi$^3$, Rui Li$^1$, 
and Yusheng Ji$^3$\\
%, Zhenyu Zhou$^4$, and \\
$^1$Yale University, \\
$^2$Peking University, \\
$^3$National Institute of Informatics, Japan
%, $^3$UCLA, $^4$Duke University
%$^1$Institute of Network Computing and Information System, School of EECS, Peking University, Beijing, China\\
%$^2$Department of Electrical Engineering, Yale University, New Haven, CT, USA\\
%$^3$Department of Computer Science, University of California, Los Angeles, CA, USA\\
%$^4$Department of Computer Science, Duke University, Durham, NC, USA\\
}
%\author{
%%\alignauthor
%%Lin Chen$^{1,2}$, Kaigui Bian$^{1}$ \\
%% , Lin Chen$^{3}$, Cong Liu$^{4}$, Jung-Min Jerry Park$^{5}$, and Xiaoming Li$^{1}$\\
%$^{1}$Peking University, Beijing, China.~~~$^{2}$Yale University, New Haven, USA.
%% ~~~$^{3}$University Paris-Sud, Orsay, France.~~~$^{3}$Sun Yat-Sen University, China.~~~$^{4}$Virginia Tech, Blacksburg, VA, USA.
%\\
%%$^{1}$\{abratchen, bkg\}@pku.edu.cn, $^{2}$lin.chen@yale.edu % , $^{3}$liucong3@mail.sysu.edu.cn, $^{4}$jungmin@vt.edu
%}

\maketitle
\begin{abstract}
The base station (BS) in a multi-channel cognitive radio (CR) network has to broadcast to secondary (or unlicensed) receivers/users on more than one broadcast channels via channel hopping (CH), because a single broadcast channel can be reclaimed by the primary
(or licensed) user, leading to \emph{broadcast failures}. Meanwhile, a secondary receiver needs to synchronize its clock with the BS's clock to avoid broadcast failures caused by the possible clock drift between the CH sequences of the secondary receiver and the BS. In this paper, we propose a CH-based broadcast protocol called \emph{SASS}, which enables a BS to successfully broadcast to secondary receivers over multiple broadcast channels via channel hopping. Specifically, the CH sequences are constructed on basis of a mathematical construct---the \emph{Self-Adaptive Skolem sequence}. Moreover, each secondary receiver under SASS is able to adaptively synchronize its clock with that of the BS without any information exchanges, regardless of any amount of clock drift. 
\end{abstract}
%\begin{IEEEkeywords}
%simplicity, beauty, elegance
%\end{IEEEkeywords}

%\vspace{-3mm}
\section{Introduction}

In an \emph{infrastructure-based (or cellular)} cognitive radio (CR)
network, the base station (BS) has to broadcast to secondary receivers/users
on more than one broadcast channels via a channel hopping (CH) process.

The \emph{broadcast failure problem} can occur for a CH-based broadcast protocol. First, the primary users (PU, or licensed users) may reclaim the spectrum band where broadcast channels reside, and the secondary receivers have
to vacate this channel according to the requirement for protection of PU. Second, there may exist a clock drift between the BS and
the secondary receiver, which can lead to broadcast failures due to the non-overlapping of their CH sequences. Note that any fast synchronization scheme requires necessary information exchange that incurs additional control overhead.

In order to address these problems, we expect the CH-based broadcast protocol to have the following properties.
\begin{enumerate}
  \item \emph{Multiple broadcast channels}. Ideally, broadcast deliveries can occur over all available broadcast channels, thus becoming invulnerable to broadcast failures caused by the PU on a single broadcast channel. % This property is termed \emph{maximum broadcast channel diversity}.
  \item \emph{Self-adaptive synchronization without information exchange}. The secondary receivers are supposed to synchronize with the BS autonomously via a self-adaptive synchronization process without any information exchange, in order to minimize both the broadcast latency and control overhead.
\end{enumerate}

% Formally, we coin the term \emph{channel hopping-based multi-channel broadcast problem} in infrastructure-based CR networks to denote the following problem: How can a base station (secondary sender) and a user (secondary receiver), given asynchronous local clocks, PU traffic, and no information exchange, achieve successful broadcast deliveries on all available broadcast channels within a bounded broadcast latency?

In this paper, we present a channel-hopping based multi-channel broadcast protocol, called
\emph{SASS}, where the CH sequences are constructed
on basis of a mathematical construct---the \emph{Self-Adaptive Skolem Sequence}. The SASS protocol has the following two noteworthy features.
\begin{itemize}
  \item The BS can successfully broadcast to secondary
receivers over multiple (or up to the maximum number of) available channels within a bounded latency, for increasing (or maximizing) the broadcast channel diversity.
  \item The broadcast latency can be minimized to be near-zero given any amount of clock drift, as each secondary receiver can adaptively synchronize its clock with the BS without any information exchange by leveraging the historical information of successful broadcast deliveries.
\end{itemize}
Our analytical and simulation results show that SASS incurs small broadcast latency and guarantees a high successful delivery rate under various network conditions.

The rest of this paper is organized as follows. We introduce the related work on channel hopping broadcast protocols and Skolem sequence in Section~\ref{sec:bkg}. We provide the system model and formulate the problem in Section~\ref{sec:problem}. In
Section~\ref{sec:protocol}, we describe the SASS broadcast protocol. We evaluate the performance of our proposed broadcast schemes in Section~\ref{sec:sim}. We conclude the paper in Section~\ref{sec:conclude}.

%\vspace{-2mm}
\section{Related Work}
\label{sec:bkg}

% \subsection{Channel-hopping based Broadcast in CR Networks}
The purpose of most CH protocols in the literature is to achieve channel rendezvous between a sender
and a receiver using jump-stay techniques~\cite{JS}, the array-based
quorum systems~\cite{Bian-TMC}, or modular arithmetics~\cite{theis}. It has not been largely investigated as yet to devise multi-channel broadcast protocols in the context of an infrastructure-based CR network.
% To the best of our knowledge, only a limited number of researches papers have so far addressed the broadcast algorithm design problem in infrastructure-based CR networks. Nevertheless,
% A vast majority of existing solutions are proposed on the basis of assumptions such as accessibility to and/or availability of all possible common channels (some protocols only support ensured delivery on a subset of common channels), or the existence of a global (network-wide) or local common broadcast/control channel.
% In addition, existing work fails to consider leveraging historical information (e.g., past successful broadcast deliveries) to reduce the broadcast latency/delay.

In~\cite{xiejiang}, a fully-distributed broadcast protocol is proposed
to provide very high successful delivery rate while achieving the
shortest broadcast delay. 
However, most of existing research on channel-hopping based broadcast protocols in CR networks~\cite{Chen-TVT-broadcast,xiejiang}
%,xiejiang-globecom,xiejiang-tvt,xiejiang-tmc,xiejiang-bracer-tmc} 
did not leverage the availability of information gathered from past successful broadcast deliveries to synchronize clocks independent of information exchange and minimize broadcast latencies. Note that existing clock synchronization techniques~\cite{JS} rely on necessary information exchange (e.g., exchange of clock time and/or relevant parameters) after a successful broadcast delivery is established.

%\vspace{-3mm}
\section{System Model} \label{sec:problem}

\textbf{Multi-channel broadcast via channel hopping}. In a CR network,
a broadcast channel may become unavailable at any time due to the
primary user's activities. Therefore, the secondary BS (or broadcast
sender) has to broadcast the content over multiple channels to ensure
successful delivery to the secondary
receiver, which we call a \emph{multi-channel broadcast} process.

Every secondary user (SU) is equipped with a single radio interface, and we use the
channel hopping sequence to define the order in which a radio
(or a SU) visits a set of broadcast channels. Suppose there are $N$
broadcast channels, labeled as $0,1,2,\ldots,N-1$. We
consider a time-slotted communication system, in which a timeslot
means the minimum time unit within which a network node accesses a
channel. Thus, the CH sequence $u_{i}$ of radio $i$ (or SU $i$)
is represented as a sequence of channel indices:
$$
u_{i}=\{u_{i}^{0},u_{i}^{1},u_{i}^{2},\ldots,u_{i}^{t},\ldots\},
$$
 where $u_{i}$ can be an infinite sequence and $u_{i}^{t}\in[0,N-1]$
represents the channel index of $u_{i}$ in the $t$-th timeslot.

\textbf{Clock drift}. Suppose the local clock of radio/SU $i$ is
$\Delta_{i}\in\mathbb{Z}$ timeslots behind the global clock; $\Delta_{i}$
can be a negative integer, and this means that the local clock of
radio/SU $i$ is in fact $-\Delta_{i}$ timeslots ahead of the global
clock. %, and according to the global clock, the timeslot 0 of radio $i$ is in fact the timeslot $\Delta_{i}$ of the global clock, and generally the timeslot $t$ of radio $i$ is in fact the timeslot $t+\Delta_{i}$ of the global clock.
From the perspective of the global clock, the CH sequence of radio
$i$ is a sequence that starts from timeslot $\Delta_{i}$:
$$
G(u_{i})=\{G(u_{i})^{\Delta_{i}},G(u_{i})^{\Delta_{i}+1},\ldots,G(u_{i})^{\Delta_{i}+t},\ldots\},
$$
 where $G(u_{i})^{\Delta_{i}+t}=u_{i}^{t},t\in\mathbb{N}\cup\{0\}$.
It is likely that the clock drift is not a multiple
of a timeslot. This is also termed that these two nodes\textquoteright{}
timeslots are  \emph{not aligned}.
However, we can extend the slot duration to be twice of the original
slot duration. We will easily arrive at the conclusion that as long
as two nodes have successful broadcast delivery on a certain channel,
the delivery duration is at least an original slot duration.

\textbf{Successful broadcast delivery}. Given two CH sequences $u_{i}$
and $u_{j}$, if there exists $t_{g}\in\mathbb{Z}$ such that $t_{g}\geq\max\{\Delta_{i},\Delta_{j}\}$
and $G(u_{i})^{t_{g}}=G(u_{j})^{t_{g}}=h$, where $h\in[0,N-1]$,
we say that a \emph{successful broadcast delivery} occurs between radios $i$
and $j$ in the $t_{g}$-th (global) timeslot on broadcast channel
$h$. The $t_{g}$-th timeslot is called a \emph{delivery slot} and
channel $h$ is called a \emph{delivery channel} between SUs~$i$
and $j$.

%\textbf{Performance metrics}. We introduce two performance metrics to evaluate the performance of such a broadcast protocol. The delivery rate between SU sender
%$i$ and receiver $j$ before global timeslot $t_{g}$ (excluded)
%is defined as
%\[
%\frac{|\{t:\max{\Delta_{i},\Delta_{j}}\leq t<t_{g},G(u_{i})^{t}=G(u_{j})^{t}\}|}{t_{g}-\max{\Delta_{i},\Delta_{j}}}.
%\]
% The delivery channel diversity between SU sender $i$ and receiver
%$j$ before global timeslot $t_{g}$ (excluded) is defined as
%\[
%|\{h\in[0,N-1]:\exists\max{\Delta_{i},\Delta_{j}}\leq t<t_{g},G(u_{i})^{t}=G(u_{j})^{t}=h\}|.
%\]

%\vspace{-3mm}
\section{An ESS-based Broadcast Protocol}

\label{sec:protocol}

\subsection{Preliminaries}
 A \emph{Skolem sequence (SS)}~\cite{skolem},
$\{\zeta_{i}\}_{0\leq i\leq2n-1}$ of order $n$, is a permutation of the sequence of $2n$ integers $\{1,1,2,2,3,3,\ldots,n,n\}$,
and it satisfies the \emph{Skolem property}:
\begin{itemize}
  \item If $\zeta_{i}=\zeta_{j},0\leq i<j\leq2n-1$,
then $j-i=\zeta_{i}+1$.
\end{itemize}
For example, the sequence $l=\{3,1,2,1,3,2\}$
is a Skolem sequence of order $n=3$. Given $i=0$ and $j=4$, we
have $\zeta_{0}=\zeta_{4}=3$, and $j-i=3+1$; given other combinations of
$i$ and $j$, the sequence $l$ also satisfies the Skolem property.
The following lemma holds \cite{skolem}.
\begin{lem}
\label{lemma_SS} A Skolem sequence of order $n$ exists if and only
if $n$ is congruent to 0 or 3 modulo 4.
\end{lem}
In \cite{skolem}, Skolem proposed a very efficient and general construction method for Skolem sequences in his proof of Theorem~2.

\textbf{Extended Skolem sequence}. We define an \emph{extended Skolem
sequence (ESS)}, $\{\zeta'_{i}\}_{0\leq i\leq2(n+1)-1}$ of order $n$,
as a permutation of the sequence of $2(n+1)$ integers: $\{0,0,1,1,2,2,\ldots,n,n\}$.

The sequence satisfies the Skolem property, i.e., if $\zeta_{i}=\zeta_{j},0\leq i<j\le2(n+1)-1$,
then $j-i=\zeta_{i}+1$. For example, the sequence $\zeta'=\{0,0,3,1,2,1,3,2\}$
is an extended Skolem sequence (ESS) of order $n=3$. It follows
immediately from Lemma~\ref{lemma_SS} that an extended Skolem
sequence of order $n$ exists if $n$ is congruent to 0 or 3 modulo 4.

Given a Skolem sequence $\{\zeta_{i}\}_{0\leq i\leq2n-1}$ of order $n$,
we can construct an extended Skolem sequence $\{\zeta'_{i}\}_{0\leq i\leq2(n+1)-1}$
of the same order by inserting two integers $\zeta'_{0}=\zeta'_{1}=0$ at the
beginning of the original SS---i.e., by letting $\zeta'_{i}=0$ when $i=0,1$;
and $\zeta'_{i}=\zeta_{i-2}$ when $1<i\leq2(n+1)-1$. In an extended Skolem
sequence $\zeta'$ of order $n$, any integer $k\in[0,n]$ appears exactly
twice in the ESS.

%\vspace{-3mm}
\subsection{ESS-based CH Sequences}

In this subsection, we use ESS to generate channel hopping sequences for the base station and secondary receivers.% which will be described in Section~\ref{sec:protocol}.

When the channel number $N$ is congruent to 0 or 1 modulo 4, then
$N-1$ is congruent to 0 or 3, and by Lemma \ref{lemma_SS}, there
exists an ESS $\{\zeta'_{i}\}_{0\leq i\leq2N-1}$ of order $N-1$. For
example, when $N=4$, the ESS-based CH sequence is
$$
\{0,0,3,1,2,1,3,2\}.
$$

When $N\not\equiv0,1\mod{4}$, we can easily use the \emph{padding
technique}\footnote{We may as well use the \emph{downsizing
technique} as an alternative. The downsizing technique means that we discard some channels so
that the new channel number $N'\leq N$ is congruent to 0 or 1 modulo
4. We only need to discard at most 2 channels.} to transform it into the case with the channel number $N'$
congruent to 0 or 1 modulo 4.
According to the
padding technique, we increase the channel number $N$ to the minimum
integer $N'$ so that $N'\geq N$ and $N'$
is congruent to 0 or 1 modulo 4. Obviously, $N\leq N'\leq N+2$. We regard the newly added $(N'-N)$
channels as aliases of the original $N$ channels\footnote{For example, if
the channel number is $3$, we add a new channel, say, Channel 4,
so that the new channel number amounts to 4. Channel 4 serves as an
alias of Channel 1.}.
  With the padding scheme, we can focus on the case where the channel
number $N'$ is congruent to 0 or 1 modulo 4.
\begin{itemize}
  \item Let $u$ and $v$ be two CH sequences of the same length, say, $T$.
We denote the set of delivery channels by $$\mathcal{C}(u,v)\triangleq\{h\in[0,N-1]:\exists t\in[0,T-1],u^{t}=v^{t}=h\}.$$
  \item Let $\mathcal{D}(u,v)$ denote the set of delivery slots between $u$
and $v$ and $$\mathcal{D}(u,v)\triangleq\{t\in[0,T-1]:u^{t}=v^{t}\}.$$
  \item We define the notion of circular shift to represent the clock drift,
i.e., $$\rt(u,\alpha)=\{w^{0},w^{1},w^{2},\ldots,w^{T-1}\}$$ is a sequence
of length $T$ and $w^{t}\triangleq u^{(t+\alpha)\bmod{T}}$.
  \item We let $$\prod_{k=1}^{K}\mu_{k}=\mu_{1}\parallel\mu_{2}\parallel\cdots\parallel\mu_{K}$$
denote the concatenation of sequences $\mu_{k}$'s.
\end{itemize}

Consider two ESS-based CH sequences with different amount of clock
drift, $\rt(u,\alpha)$ and $\rt(u,\beta)$, where $u$ is an ESS. Their
relative clock drift is $\alpha-\beta$.
\begin{itemize}
\item Theorem \ref{thm:simple} shows the relationship between the set of
broadcast delivery channels $\mathcal{C}(\rt(u,\alpha),\rt(u,\beta))$
and the relative clock drift $(\alpha-\beta)$.
\item Theorem \ref{thm:density} shows the relationship between the set
of broadcast delivery slots $\mathcal{D}(\rt(u,\alpha),\rt(u,\beta))$
and the relative clock drift $(\alpha-\beta)$.
\end{itemize}
According to these two theorems, a secondary receiver is able to figure
out the exact clock drift between the BS and the receiver itself by
simply looking at the historical information---i.e., the set of broadcast
delivery channels---such that the receiver can synchronize to the
BS without exchanging any control messages.
\begin{thm}
\label{thm:simple}Suppose that $u=\{\zeta'_{i}\}_{0\leq i\leq2N'-1}$ is an ESS of
order $N'-1$.
\begin{enumerate}
\item   $\mathcal{C}(\rt(u,\alpha),\rt(u,\beta))=\{0,1,2,\ldots,N'-1\}$ if and only if $\alpha-\beta\equiv0\pmod{2N'}$.
\item  $\mathcal{C}(\rt(u,\alpha),\rt(u,\beta))=\{|g|-1\}$, where $|g|\leq N'$, if and only if $\alpha-\beta\equiv g\not\equiv0\pmod{2N'}$.
\end{enumerate}
\end{thm}
\begin{proof}
The first clause is obvious. Now it suffices to show that $$\alpha-\beta\equiv g\not\equiv0\pmod{2N'}$$ will imply $$\mathcal{C}(\rt(u,\alpha),\rt(u,\beta))=\{|g|-1\},$$
where $|g|\leq N'$. If $g>0$, then $\exists0\leq i_{0}<j_{0}<2N'$ s.t.
$u^{i_{0}}=u^{j_{0}}=g-1$. Thus $j_{0}-i_{0}=g$, $\rt(u,g)^{i_{0}}=u^{i_{0}+g}=u^{i_{0}}=g-1$,
i.e., $g-1\in\mathcal{C}(\rt(u,g),u)$.
Suppose $x\in\mathcal{C}(\rt(u,g),u)$,
then $\exists0\leq i_{0}<2N'$ s.t. $\rt(u,g)^{i_{0}}=u^{i_{0}+g}=u^{i_{0}}=x$.
By the definition of ESS, $x+1=(i_{0}+g)-i_{0}=g$, $x=g-1$. Therefore
$$\mathcal{C}(\rt(u,\alpha),\rt(u,\beta))=\mathcal{C}(\rt(u,g),u)=\{|g|-1\}.$$
If $g<0$, then $\beta-\alpha\equiv-g\not\equiv0\pmod{2N'}$. Thus $\mathcal{C}(\rt(u,\beta),\rt(u,\alpha))=\{|-g|-1\}=\{|g|-1\}$.\end{proof}

\begin{thm}
\label{thm:density} Suppose that $u=\{\zeta'_{i}\}_{0\leq i\leq2N'-1}$ is an ESS of
order $N'-1$.
\begin{enumerate}
\item $\mathcal{D}(\rt(u,\alpha),\rt(u,\beta))=\{0,1,2,3,\ldots,2N'-1\}$ if and only if $\alpha-\beta\equiv0\pmod{2N'}$.
\item If $\alpha-\beta\equiv g\pmod{2N'}$, where $|g|\leq N'$, then

\begin{enumerate}
\item $|\mathcal{D}(\rt(u,\alpha),\rt(u,\beta))|=1$ if and only if $0<|g|<N'$.
\item $|\mathcal{D}(\rt(u,\alpha),\rt(u,\beta))|=2$ if and only if $|g|=N'$.
\end{enumerate}
\end{enumerate}
\end{thm}
\begin{proof}
The first clause is obvious. Now suppose $\alpha-\beta\equiv g\pmod{2N'}$
($|g|\leq N'$). It suffices to show that $0<|g|<N'\Rightarrow|\mathcal{D}(\rt(u,\alpha),\rt(u,\beta))|=1$
and $|g|=N'\Rightarrow|\mathcal{D}(\rt(u,\alpha),\rt(u,\beta))|=2$.
By Theorem \ref{thm:simple}, $|\mathcal{D}(\rt(u,\alpha),\rt(u,\beta)|\geq1$
and $\mathcal{C}(\rt(u,\alpha),\rt(u,\beta))=\{|g|-1\}$, i.e., the delivery
channel $h=|g|-1$ and $\exists0\leq i_{0}<2N'$ s.t. $\rt(u,\alpha)^{i_{0}}=\rt(u,\beta)^{i_{0}}=h$.
There are only two $h$'s in $\rt(u,\alpha)$ and $\rt(u,\beta)$, respectively.
Without loss of generality, the remaining $h$ in $\rt(u,\alpha)$ is
$\rt(u,\alpha)^{\left[i_{0}-(h+1)\right]\bmod{2N'}}$ and that in $\rt(u,\beta)$
is $\rt(u,\beta)^{\left[i_{0}+(h+1)\right]\bmod{2N'}}$. If $i_{0}+(h+1)\equiv i_{0}-(h+1)\pmod{2N'}$,
we have $2(h+1)\equiv2|g|\equiv0\pmod{2N'}$. Since $g\neq0$, we
conclude that $|g|=N'$.
Therefore if $|g|<N'$, $i_{0}+(h+1)\not\equiv i_{0}-(h+1)\pmod{2N'}$
and $|\mathcal{D}(\rt(u,\alpha),\rt(u,\beta))|=1$; if $|g|=N'$, then
$i_{0}+(h+1)\equiv i_{0}-(h+1)\pmod{2N'}$ and $|\mathcal{D}(\rt(u,\alpha),\rt(u,\beta))|=2$.
\end{proof}
%\vspace{-3mm}
\subsection{SASS: A Two-phase Broadcast Protocol}

The broadcast CH sequence for the BS (sender) is $\prod_{n=0}^{\infty}\mu$,
where $\mu$ is a prespecified ESS of order $N'-1$ and thus with a
length of $2N'$. %Note that $\mu$ is known to all nodes in this network.

Every secondary receiver has an initial CH sequence. Then, the CH
sequence is dynamically updated by the receiver based on the historical
information of whether successful broadcast delivery occurs in the
past timeslots. Specifically, it calibrates the clock in a self-adaptive
manner by leveraging the mathematical properties of ESS, so as to
synchronize to the BS (sender) for minimizing the broadcast latency.

\textbf{Phase~1: ESS-based CH sequence generation}. A secondary receiver
initially uses the CH sequence
$$
\prod_{n=0}^{\infty}\rt(\mu,n)
$$
until the first successful broadcast delivery occurs. Meanwhile,
it counts/maintains the number of successful broadcast deliveries
since its local clock's timeslot $\lfloor\frac{t}{2N'}\rfloor\cdot2N'$,
where $t$ is the current timeslot according to its local clock.

The secondary receiver will not wait long for the first successful
broadcast delivery to occur. If the PU signal is not present in all channels,
the first successful broadcast delivery will occur within $4N'(N'-1)$
timeslots after both the secondary sender and receiver start channel
hopping, as shown by Theorem~\ref{thm:firstbdst}.

\begin{thm}
  \label{thm:firstbdst}
Under SASS, the first successful broadcast delivery occurs within $4N'(N'-1)$ timeslots after both the secondary sender and receiver start channel hopping.
\end{thm}
\begin{proof}
Suppose the sender uses the ESS $u$ of length $2N'$, and the receiver uses $u$, $\rt(u,1)$, $\rt(u,2)$, $\rt(u,3),\ldots$, $\rt(u,2{N'}-1)$, sequentially. Since the receiver exhausts all possible cyclic rotations of $u$, the first successful broadcast delivery will occur within $2N'\cdot 2N'=4{N'}^2$ slots.

Then, we can further improve this upper bound to $4N'(N'-1)$. Now we establish our argument for the case $N'=4$ as an example. For a general $N'$, we can have a similar argument.

Suppose $N'=4$. From the perspective of the sender's clock, the receiver uses $\rt(u,a)$, $\rt(u,a+1)$, $\rt(u,a+2),\ldots$, $\rt(u,a+7)$, where $a$ is the clock drift between the sender and the receiver. Now we consider the consecutive $4{N'}^2=64$ slots.

\begin{table}

\caption{\textnormal{The cases of broadcast channels if the sender uses $u$ and the receiver uses $\rt(u,a)$ with a clock drift of $a$ slots, where $a$ varies from 0 to $2N'-1=7$, given $N'=4$. The ``All'' here means that if the sender and receiver are synchronized, they can have broadcast delivery on all channels. The ``0'' here means that if the clock drift between the sender and receiver is 1, they can have broadcast delivery on Channel 0, and similarly for ``1'', ``2'' and ``3'' hereinafter. \label{tab:sdrrvr}}}
\centering

\begin{tabular}{|c|c||c|c|}
\hline
Rx & Bdcast ch(s) & Rx & Bdcast ch(s) \\
\hline
     $\rt(u,0)$       &        All &    $\rt(u,1)$       &          0 \\
\hline
   $\rt(u,2)$       &          1  &    $\rt(u,3)$       &          2 \\
\hline
     $\rt(u,4)$      &          3 &     $\rt(u,5)$      &          2 \\
\hline
   $\rt(u,6)$       &          1  &    $\rt(u,7)$       &          0 \\
\hline
\end{tabular}
\end{table}

Table~\ref{tab:sdrrvr} shows the cases of broadcast channels for all possible clock drifts. If the sender and the receiver are synchronized, the receiver starts from the second row (i.e., it starts from $\rt(u,0)=u$), and then uses $\rt(u,1)$, $\rt(u,2)$, $\rt(u,3),\ldots$, $\rt(u,7)$, and then back to $\rt(u,0)=u$ and so on. If they have a clock drift of $a$ slots, the receiver uses $\rt(u,a)$, $\rt(u,a+1)$, $\rt(u,a+2),\ldots$, $\rt(u,a+7)$, then back
to $\rt(u,a)$ and so on. If $a=0$, broadcast delivery occurs on all channels, which is the best case. The worst case is when the receiver starts from $\rt(u,1)$. In this case, when it arrives at $\rt(u,6)$, successful broadcast delivery must have occurred because from $\rt(u,1)$ to $\rt(u,6)$, they have tried all possible channels 0, 1, 2, and 3. Thus first successful delivery will occur within $2\times 4\times (2\times 4-2)$ slots. For a general $N'$, the upper bound is $2{N'} (2{N'}-2)=4N'({N'}-1)$.

\end{proof}

\textbf{Phase~2: Self-adaptive clock calibration}. Upon the first
successful broadcast delivery, the secondary receiver enters the clock
calibration phase: it knows the exact clock drift between the sender
(BS) and itself, and then adaptively synchronize to the sender.

According to the local clock of the secondary receiver, we group every
$2N'$ timeslots into a (time) frame, e.g. timeslots 0 to $2N'-1$
of the receiver form the first frame. Suppose SU~$i$ is the sender,
SU~$j$ is the receiver, and that the first broadcast delivery occurs
in channel $\alpha\in[0,N'-1]$ in timeslot $\tau_{1}$ that lies
in the $\phi$-th frame according to SU~$j$'s clock.

By the definition of ESS, there exists another timeslot $\tau_{2}\neq\tau_{1}$
that lies in the $\phi$-th frame and contains channel $\alpha$.
Let $u_{j}[\phi]$ denote the segment of SU $j$'s CH sequence in
the $\phi$-th frame according to SU $j$'s local clock, and we have
$u_{j}[\phi]=\rt(\mu,\phi-1)$.

In the calibration phase, the receiver (i.e. SU $j$) continues using
its original CH sequence $\prod_{n=0}^{\infty}\rt(\mu,n)$ until
the end of its $\phi$-th frame; and in the meanwhile, SU $j$ continues
counting the number of successful broadcast deliveries since its local
clock's timeslot $\lfloor\frac{t}{2N'}\rfloor\cdot2N'=2N'(\phi-1)$,
where $t$ is the current timeslot according to its local clock.

When the $\phi$-th frame ends, the receiver checks if successful
broadcast delivery occurs in timeslot $\tau_{2}$ and obtains the
total number of successful broadcast deliveries in the $\phi$-frame,
which is denoted by $SB[\phi]$. The secondary receiver chooses different
CH sequences in the following cases.

\emph{Case~1}: Successful broadcast delivery occurs in timeslot $\tau_{2}$
and $\alpha\neq N'-1$. In this case, the receiver knows that the
segment of SU $i$'s CH sequence in the $\phi$-th frame according
to SU $j$'s clock, denoted by $u_{i}[\phi]$, is exactly $u_{j}[\phi]=\rt(\mu,\phi-1)$.
From the next frame (the $(\phi+1)$-th frame) on, the receiver uses
$$
\prod_{n=0}^{\infty}\rt(\mu,\phi-1)
$$
as its new CH sequence.

Fig.~\ref{fig:case1} shows an example of Case~1 with $N'=4$. The
PU occupies channels 0 and 3. The first successful delivery occurs
on channel 1 (thus $\alpha=1\neq N'-1$) and in the 4th slot of the
$\phi$-th frame. Slot~$\tau_{2}$ is the 6th slot of frame $\phi$
and successful delivery occurs in slot $\tau_{2}$. Thus the receiver
(SU $j$) knows that it has been synchronized with the sender (SU
$i$). So it continues using $\rt(\mu,\phi-1)$ as its CH sequence.

\emph{Case~2}: $\alpha=N'-1$. This implies that successful broadcast
delivery occurs in timeslot $\tau_{2}$; and that $u_{i}[\phi]$ is
either $u_{j}[\phi]$ or $\rt(u_{j}[\phi],N')$.
\begin{itemize}
\item In the next frame (the $(\phi+1)$-th frame), the receiver uses $\rt(u_{j}[\phi],N')$
as its CH sequence and counts the number of successful broadcast deliveries
in the $(\phi+1)$-th frame, denoted by $SB[\phi+1]$.
\item From the $(\phi+2)$-th frame on, the receiver knows the exact difference
between its clock and the sender's, and it can choose the CH sequence
synchronized with the sender. If $SB[\phi]\geq SB[\phi+1]$, it chooses
$$
\prod_{n=0}^{\infty}u_{j}[\phi]=\prod_{n=0}^{\infty}\rt(\mu,\phi-1);
$$
otherwise, it uses $$
\prod_{n=0}^{\infty}u_{j}[\phi+1]=\prod_{n=0}^{\infty}\rt(\mu,\phi-1+N').
$$
\end{itemize}
Fig.~\ref{fig:case2} shows an example of Case~2 with $N'=4$. The
PU occupies channels 1 and 2. The first successful delivery occurs
on channel 3 (thus $\alpha=3=N'-1$) and in the 3th slot of the $\phi$-th
frame. Slot $\tau_{2}$ is the 7th slot of frame $\phi$ and successful
delivery occurs in slot $\tau_{2}$. Thus the receiver (SU $j$) knows
that the segment of the sender's CH sequence in frame $\phi$ is either
$u_{j}[\phi]=\{2,1,3,2,0,0,3,1\}$ or $\rt(u_{j}[\phi],4)=\{0,0,3,1,2,1,3,2\}$.
So it tries $\rt(u_{j}[\phi],4)=\{0,0,3,1,2,1,3,2\}$ in frame
$(\phi+1)$ and finds that $SB[\phi]=2$, $SB[\phi+1]=4$ and $SB[\phi]<SB[\phi+1]$.
Therefore it knows that the sender uses $\rt(u_{j}[\phi],4)=\{0,0,3,1,2,1,3,2\}$.
So it synchronizes its clock with the sender and uses $\{0,0,3,1,2,1,3,2\}$
as its CH sequence from frame~$(\phi+2)$ on.

\emph{Case~3}: Successful broadcast delivery does not occur in timeslot
$\tau_{2}$. This implies that $\alpha\neq N'-1$ and that either
$u_{i}[\phi]=\rt(u_{j}[\phi],\alpha+1)$ or $u_{i}[\phi]=\rt(u_{j}[\phi],-(\alpha+1))$.
\begin{itemize}
\item The receiver uses $\rt(u_{j}[\phi],\alpha+1)$ and $\rt(u_{j}[\phi],-(\alpha+1))$
as its CH sequences in the $(\phi+1)$-th and $(\phi+2)$-th frames;
it counts the numbers of successful broadcast deliveries in these
two frames, denoted by $SB[\phi+1]$ and $SB[\phi+2]$, respectively.
\item From the $(\phi+3)$-th frame on, if $$SB[\phi+1]\geq SB[\phi+2],$$
the receiver chooses $$
\prod_{n=0}^{\infty}u_{j}[\phi+1]=\prod_{n=0}^{\infty}\rt(\mu,\phi-1+(\alpha-1));
$$
otherwise, it uses $$
\prod_{n=0}^{\infty}u_{j}[\phi+2]=\prod_{n=0}^{\infty}\rt(\mu,\phi-1-(\alpha-1)).
$$
\end{itemize}
Fig.~\ref{fig:case3} shows an example of Case~3 with $N'=4$. The
PU occupies channels 2 and 3. The first successful delivery occurs
on channel 1 (thus $\alpha=1\neq N'-1$) and in the 6th slot of the
$\phi$-th frame. Slot $\tau_{2}$ is the last slot of frame $\phi$
and successful delivery does not occur in slot $\tau_{2}$. Thus the
receiver (SU $j$) knows that the segment of the sender's CH sequence
in frame $\phi$ is either $\rt(u_{j}[\phi],2)=\{0,0,3,1,2,1,3,2\}$
or $\rt(u_{j}[\phi],-2)=\{2,1,3,2,0,0,3,1\}$. So it tries $\rt(u_{j}[\phi],2)=\{0,0,3,1,2,1,3,2\}$
in frame $(\phi+1)$ and tries $\rt(u_{j}[\phi],-2)=\{2,1,3,2,0,0,3,1\}$
in frame $(\phi+1)$. It finds that $SB[\phi+1]=4$, $SB[\phi+2]=0$
and $SB[\phi+1]>SB[\phi+2]$. Thus it knows that the sender uses $\rt(u_{j}[\phi],2)=\{0,0,3,1,2,1,3,2\}$.
So it synchronizes its clock with the sender and uses $\rt(u_{j}[\phi],2)=\{0,0,3,1,2,1,3,2\}$
as its CH sequence from frame~$(\phi+3)$ on.

After completing the clock calibration, every secondary receiver is
synchronized to the BS. Since the BS and secondary receivers use the
same ESS-based CH sequence, the broadcast latency will be minimized
to zero upon the successful clock calibration in Phase~2.

\begin{figure}[t]
 \centering \subfigure[An example of Case 1 with $N'=4$. The PU occupies channels 0 and
3. Therefore the grids with 0 and 3 inside are gray.]{ \label{fig:case1} \includegraphics[height=1.3cm]{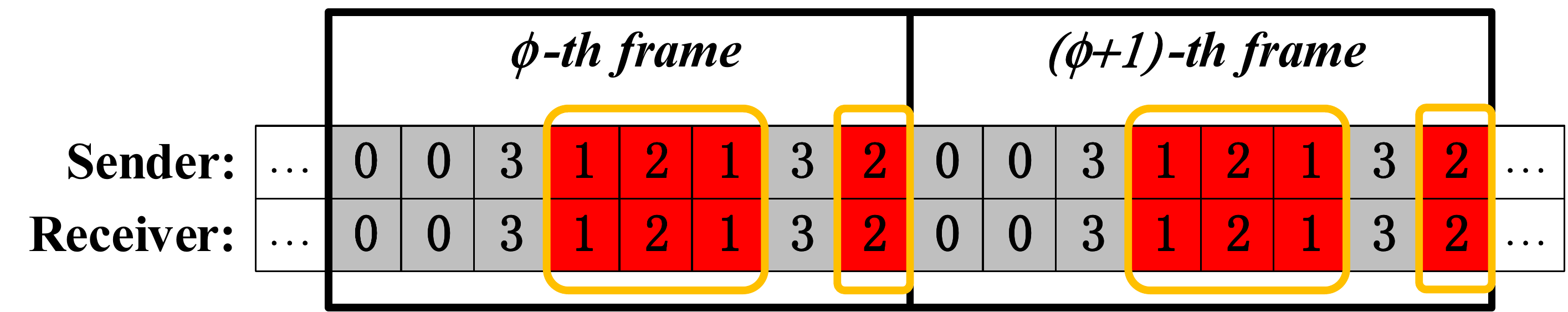} }
\hfill{}\subfigure[An example of Case 2 with $N'=4$. The PU occupies channels 1 and
2. Therefore the grids with 1 and 2 inside are gray.]{ \label{fig:case2}\includegraphics[height=1.3cm]{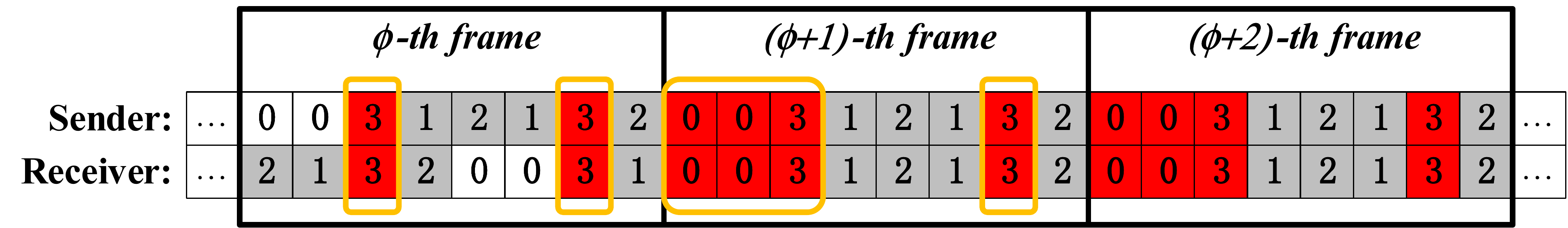} } \hfill{}\subfigure[An example of Case 3 with $N'=4$. The PU occupies channels 2 and
3. Therefore the grids with 2 and 3 inside are gray.]{ \label{fig:case3}\includegraphics[height=1.3cm]{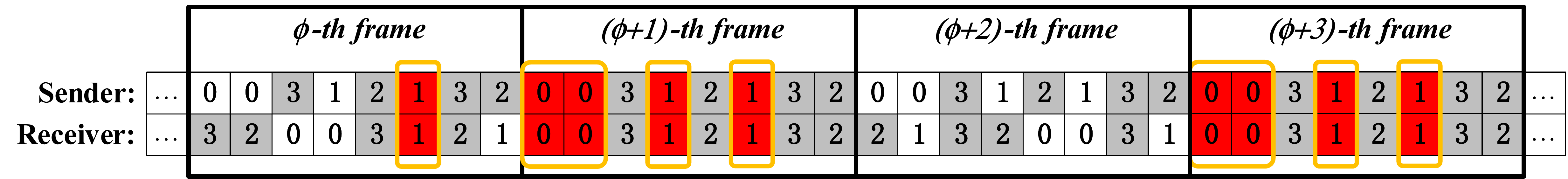} } \caption{Examples illstrating the three cases. Each grid denotes a timeslot.
The number in a grid denotes the index of a channel that the node
hops to in this timeslot. A gray grid represents that the PU occupies
this channel in the timeslot, while a non-gray one represents no PU
presence on this channel in the timeslot. A red grid indicates that
successful broadcast delivery occurs in the timeslot and over this
channel.}

\label{fig:3cases}
\end{figure}

%\vspace{-2mm}
\section{Performance Evaluation} \label{sec:sim}
\subsection{Simulation Setup}
In this section, we compare the performance of the proposed broadcast
protocol SASS and existing CH based broadcast protocols via simulation results: the
random channel hopping broadcast protocol (RCH), the canonical
Skolem sequence based broadcast protocol without self-adaptivity (CSS) and the asynchronous channel hopping protocol (ACH) \cite{Bian-TMC}.

In each simulated broadcast pair (secondary sender/BS and secondary receiver), each
node can access $N$ broadcast channels (i.e., the number of broadcast
channels available to the broadcast pair is $N$). And both of the
two nodes generate their CH sequences using the agreed broadcast protocol
(i.e., either the proposed SASS protocol or other existing broadcast
protocols) and perform channel hopping in accordance with the sequences.

\textbf{Primary user traffic}. We simulated a number of $X$ primary
transmitters operating on $X$ channels independently, and these channels
were randomly chosen in each simulation run. In most existing work,
it is assumed that a primary user transmitter follows a ``busy/idle\textquotedblright{}
transmission pattern on a licensed channel \cite{Geirhofer,Huang},
and we assume the same traffic pattern here --- i.e., the busy period
has a fixed length of $b$ timeslots, and the idle period follows
an exponential distribution with a mean of $l$ timeslots. A channel
is considered \emph{unavailable} when PU signals are
present in it. The intensity of primary user traffic can be characterized
as $PU=\frac{X}{N}\cdot\frac{b}{l+b}\times100\%$.

\textbf{Random clock drift}. In a CR network, the nodes may lose clock
synchronization or even link connectivity at any time when they experience
the broadcast failure problem due to primary user activities. Hence,
the clock of the nodes are not necessarily synchronized. In each simulation
run, each secondary node determines its clock time independently of
other nodes.
%In this section, we compare the proposed broadcast protocol SASS and
%a few channel hopping protocols: random channel hopping (RCH), the
%canonical Skolem sequence based channel hopping (CSS), and the asynchronous
%channel hopping (ACH)~\cite{Bian-TMC} protocols via simulation results.
%Each simulated broadcast node pair (secondary sender and receiver)
%can access up to $N$ broadcast channels. We simulated a number of
%$X$ primary transmitters operating on $X\geq N$ channels independently,
%and these channels were randomly chosen in each simulation run. A
%primary user transmitter follows a ``busy/idle\textquotedblright{}
%transmission pattern on a licensed channel \cite{Geirhofer,Huang}.
%That is, the busy period has a fixed length of $b$ timeslots, and
%the idle period follows an exponential distribution with a mean of
%$l$ timeslots. The intensity of primary user traffic can be characterized
%as $PU=\frac{X}{N}\cdot\frac{b}{l+b}\times100\%$. Moreover, each
%secondary node determines its clock time independently of the global
%clock, thus leading to a random clock drift between the sender and
%receiver nodes.

%\textcolor{red}{In fig 2, change Y-axis to \textquotedbl{}Prop. of bdcast slots\textquotedbl{}. fig 2-(d), the label of its Y-axis is incomplete, check the original pdf image. Change \textquotedbl{}asym ach\textquotedbl{} to \textquotedbl{}ach\textquotedbl{} in all legends.}

%\vspace{-3mm}
\subsection{Proportion of Successful Broadcast Slots}

We define the proportion
of successful delivery slots in the first $t$ timeslots, $\rho(t)$,
as the percentage of timeslots in the first $t$ timeslots in which successful
broadcast delivery occurs.

Figs.~\ref{fig:pu0}, \ref{fig:pu25}, \ref{fig:pu50} and \ref{fig:pu75}
illustrate the results given the PU traffic $PU=0\%$, $25\%$, $50\%$
and $75\%$, respectively. In the proposed SASS protocol, the proportion
of successful broadcast slots progressively approximates to the theoretical
maximum $1-PU$---the proportion values are $100\%$, $75\%$, $50\%$, and $25\%$ respectively. However, the performance of other protocols is approximately
stable at $\frac{1-PU}{N}$.

% In sum, because of its self-adaptive property, our proposed SASS has a significant gain in delivery rate.

\begin{figure}[htbp]
\centering
 \subfigure[$PU=0\%$;]{ \label{fig:pu0}\includegraphics[width=1.8in]{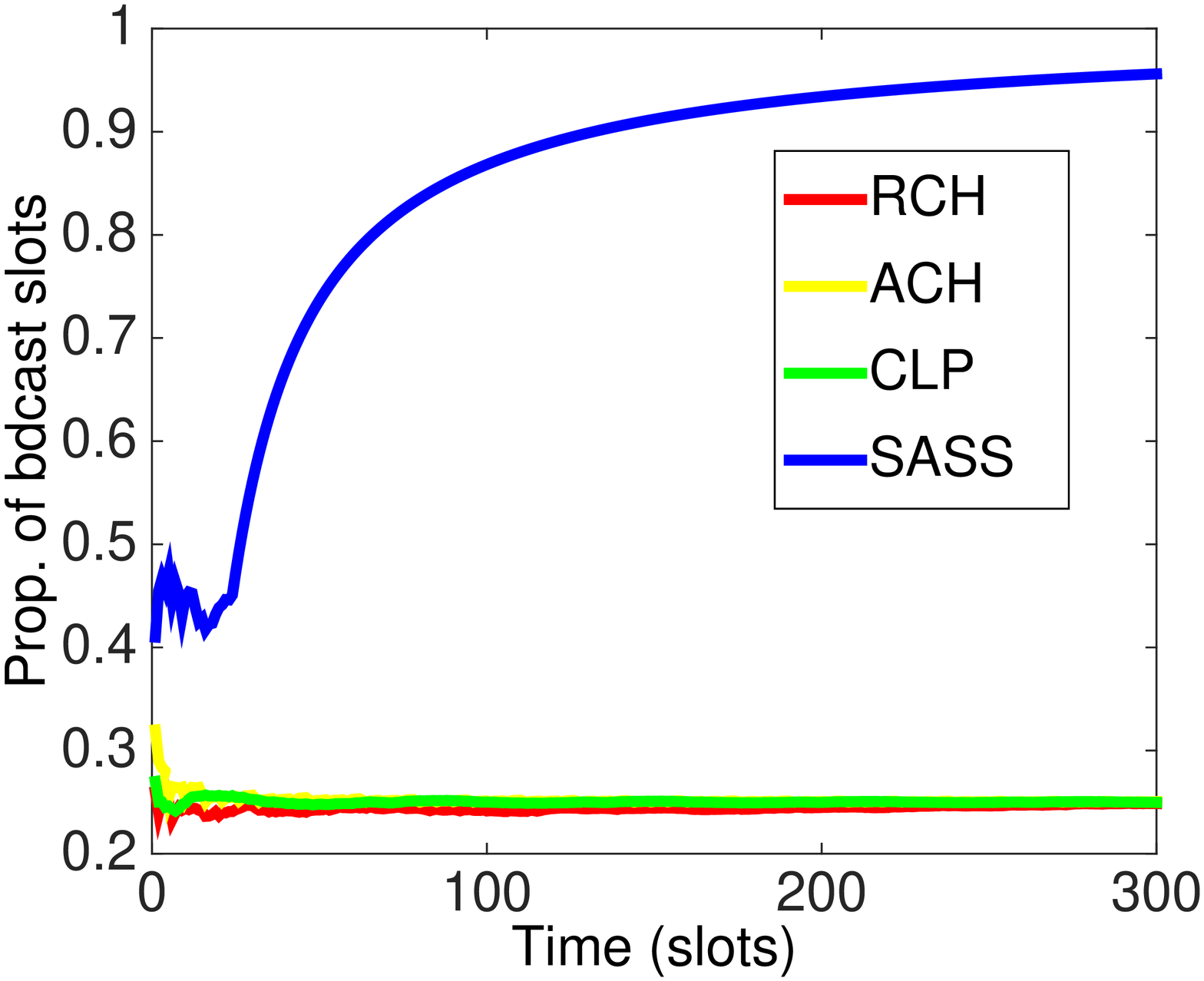}}
\hfill{}\subfigure[$PU=25\%$;]{ \label{fig:pu25}\includegraphics[width=1.8in]{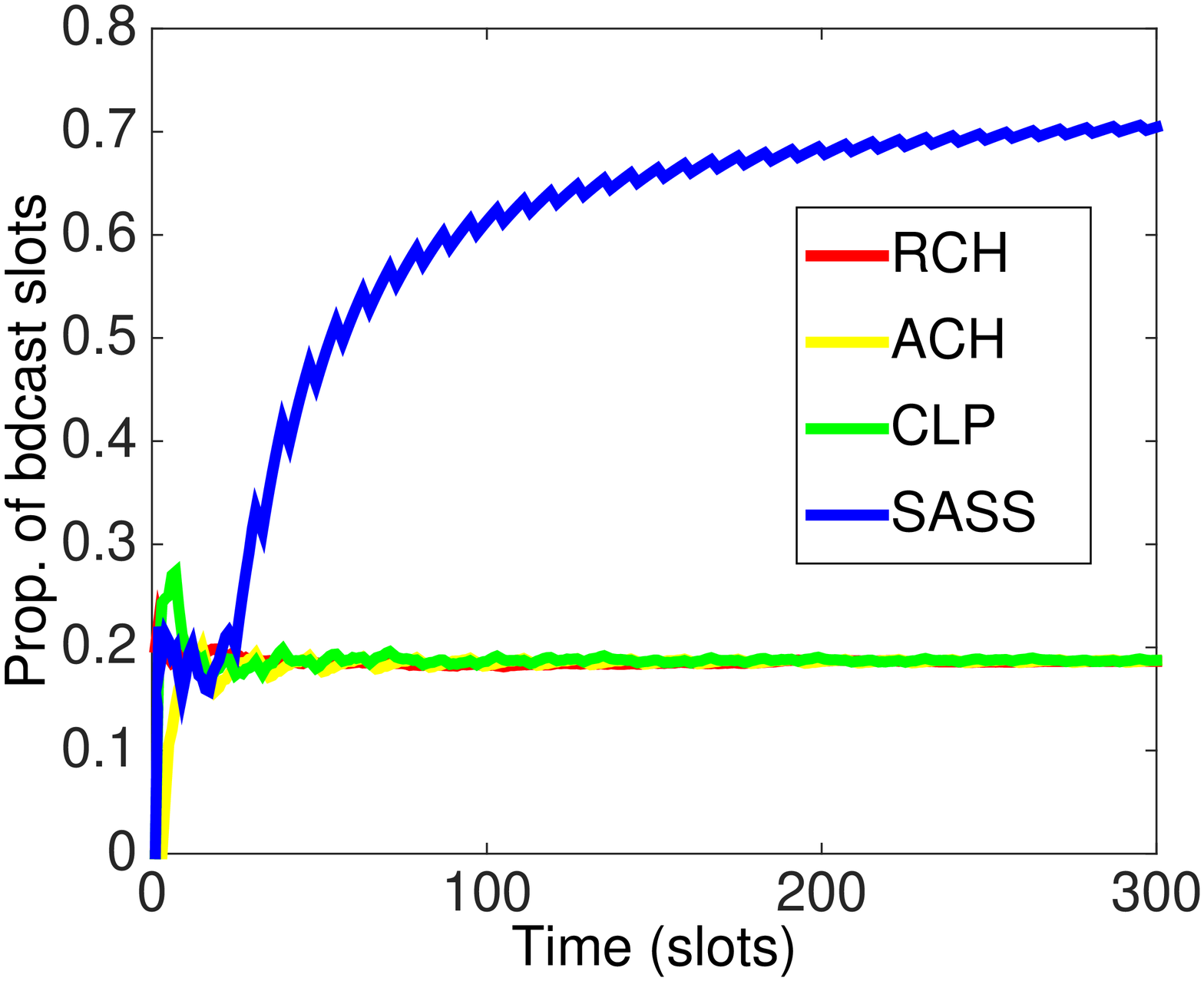}}
\hfill{}\subfigure[$PU=50\%$;]{ \label{fig:pu50}
\includegraphics[width=1.8in]{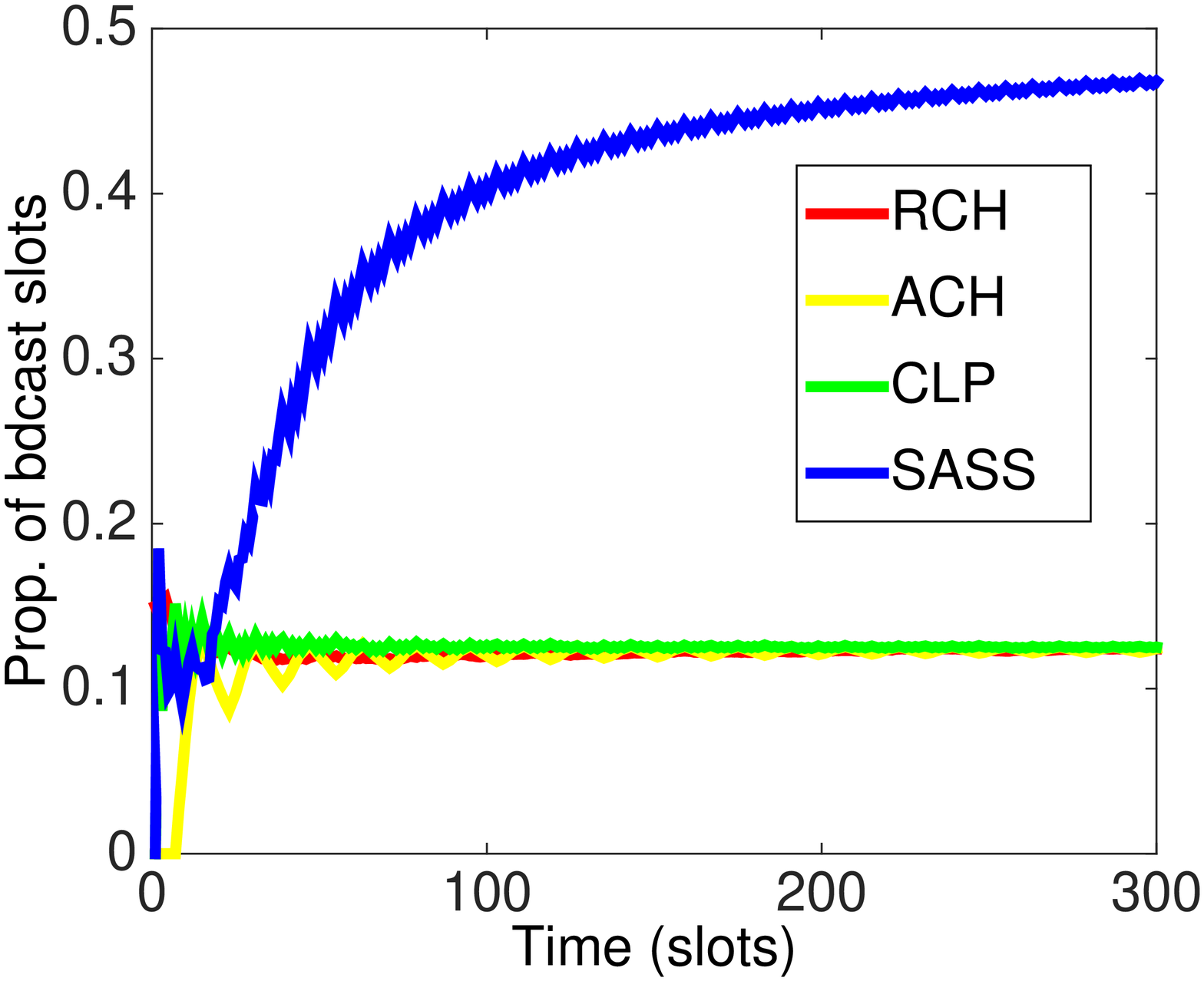}}
\hfill{}\subfigure[$PU=75\%$.]{ \label{fig:pu75}\includegraphics[width=1.8in]{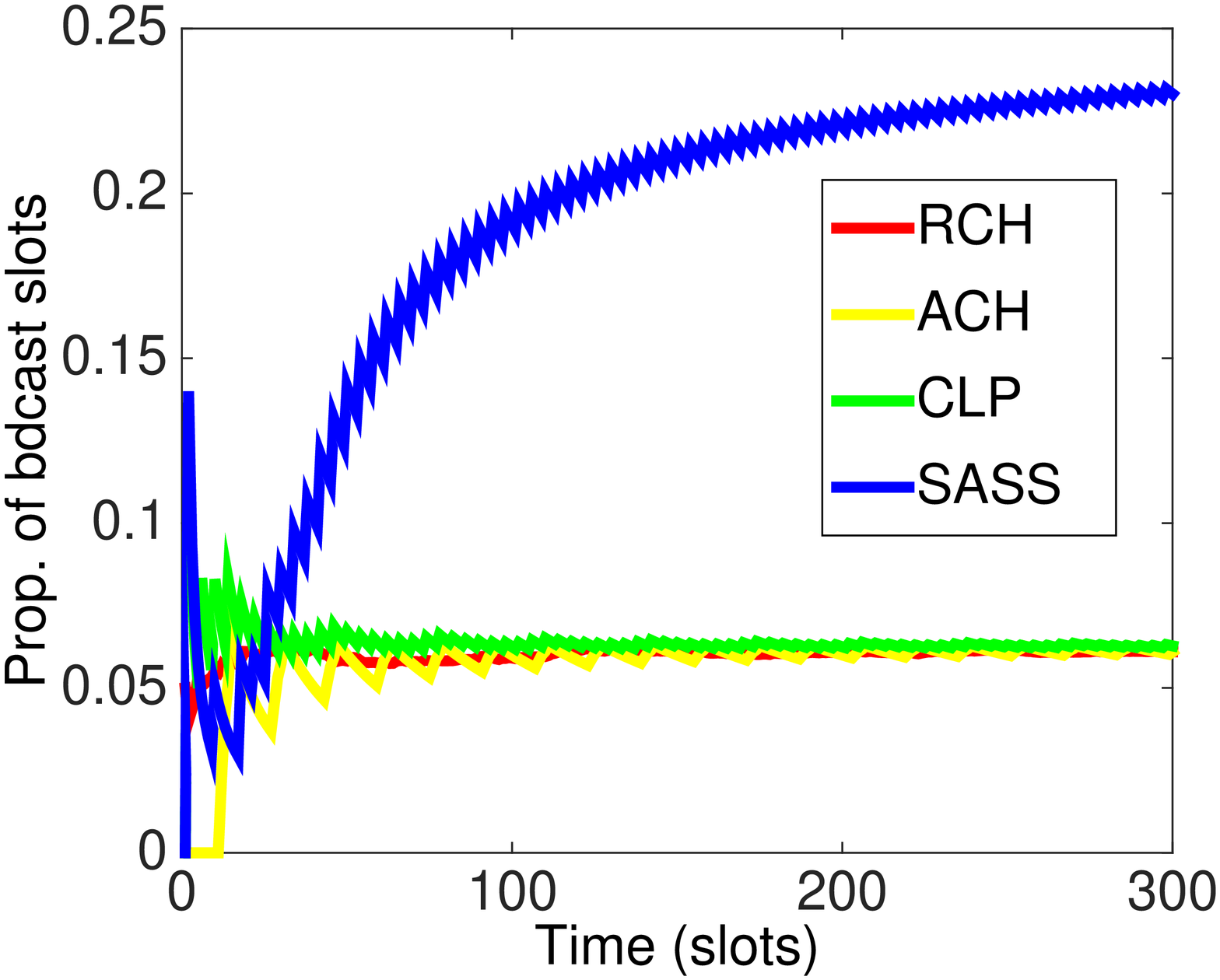}}
\caption{Proportion of successful broadcast slots vs. time. \label{fig:rate}}
\end{figure}

%\vspace{-2mm}
\subsection{Broadcast Latency}%We show the broadcast latency until the
%first one/50/100/150/200 successful broadcast deliveries occur in
%Fig.~\ref{fig:latency}. The latency under the SASS protocol progressively
%outperforms the other three protocols as the number of successful
%broadcast deliveries increases. This is because the SASS protocol
%can synchronize all of the receivers with the broadcast sender, thus
%greatly reducing the delivery latency on average.

In this set of simulations, we simulate $1000$ pairs of nodes, and investigate the broadcast latencies under the proposed SASS and other
existing CH based broadcast protocols in the following five scenarios: (1) the
latency until the first successful broadcast delivery occurs; and
the average delivery latency in the first (2) 50, (3) 100, (4) 150, and
(5) 200 timeslots. The results are showed in Fig.~\ref{fig:latency}.

We observe that the latency under the SASS protocol progressively
outperforms the other three protocols as the number of successful
broadcast deliveries increases. Its delivery latency drops
down to 7 and then 5 and finally decreases below 5, while the other
three protocols' latency remains above 15. This can be attributed to the fact that the SASS protocol can synchronize all of the receivers with the broadcast sender, thus
greatly reducing the delivery latency on average.

\begin{figure}
\centering
\includegraphics[width=3in]{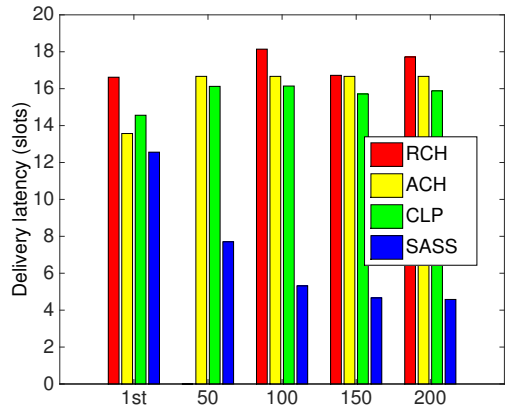}\caption{The average delivery latency in the following five scenarios: the latency until the first successful broadcast delivery occurs (see the leftmost group of bar labeled ``1st''); and the average delivery latency
in the first 50, 100, 150, and 200 timeslots (please see the groups of bars labeled
``50, 100, 150, and 200'', respectively). \label{fig:latency}}
\end{figure}

%\vspace{-2mm}
\section{Conclusion} \label{sec:conclude}

In this paper, we propose a channel hopping based multi-channel broadcast
protocol, called SASS, where the CH sequences are constructed on basis
of the self-adaptive extended Skolem sequence. SASS allows the network
base station (broadcast sender) to broadcast over multiple channels
such that the broadcasts can be successfully delivered to secondary
receivers. Meanwhile, each secondary receiver can infer the difference
between its clock and the clock of the sender, and then adaptively
synchronize with the sender to further reduce the broadcast latency.
SASS is robust to the broadcast failure caused by primary user activities
and the clock drift between two CH sequences.

%Our analytical and simulation results show that the proposed solution
%achieves the theoretical maximum delivery rate, reduces the broadcast
%latency and achieves full broadcast channel diversity and it is robust
%to the broadcast failure (broadcast link breakage) caused by primary
%user activities under various network conditions.


\begin{thebibliography}{99}
%\bibitem{technical-report}A Multi-channel Broadcast Protocol in Cognitive
%Radio Networks (technical report). \url{http://techicalrep.weebly.com/uploads/2/2/7/5/22754626/tr-broadcast.pdf}.

\bibitem{Bian-TMC}K. Bian, and J.-M. Park, ``Maximizing rendezvous
diversity in rendezvous protocols for decentralized cognitive radio
networks,'' \emph{IEEE Transactions on Mobile Computing}, vol. 12,
no. 7, pp. 1294--1307, 2013.

\bibitem{Chen-TVT-broadcast}L. Chen, K. Bian, X. Du, and X. Li. ``Multi-channel Broadcast via Channel Hopping in Cognitive Radio Networks,'' \emph{IEEE Transactions on Vehicular Technology}, vol. 64, no. 7, pp. 3004--3017, Jul. 2015.

\bibitem{Geirhofer}S. Geirhofer, L. Tong, and B. Sadler. ``Cognitive
medium access: constraining interference based on experimental models,\textquotedblright{}
\emph{IEEE J. Selected Areas of Comm.}, vol. 26, no. 1, pp. 95--105,
Jan. 2008.

\bibitem{Huang}S. Huang, X. Liu, and Z. Ding, ``Optimal transmission
strategies for dynamic spectrum access in cognitive radio networks,\textquotedblright{}
\emph{IEEE Transactions on Mobile Computing}, vol. 8, no. 12, pp.
1636--1648, Dec. 2009.

%\bibitem{Skolem}C. D. Langford, Problem. \emph{The Mathematical
%Gazette}, vol. 42, no. 341, p. 228, Oct. 1958.

\bibitem{JS}Z. Lin, H. Liu, X. Chu, and Y.-W. Leung, ``Jump-stay
based channel hopping algorithm with guaranteed rendezvous for cognitive
radio networks,\textquotedblright{} in \emph{Proc. IEEE INFOCOM},
2011, pp. 2444--2452.

\bibitem{skolem}T. Skolem, ``On certain distributions of integers
in pairs with given differences,'' \emph{Mathematica Scandinavica},
vol. 5, pp. 57--68, 1957.

\bibitem{xiejiang}Y. Song and J. Xie, ``A distributed broadcast
protocol in multi-hop cognitive radio ad hoc networks without a common
control channel,'' in \emph{Proc. IEEE INFOCOM}, 2012, pp. 2273--2281.

%\bibitem{xiejiang-globecom}Y. Song and J. Xie, ``A QoS-based broadcast protocol for multi-hop cognitive radio ad hoc networks under blind information,'' in \emph{Proc. of IEEE Global Telecommunications Conference (GLOBECOM)}, 2011, pp. 1--5.
%
%\bibitem{xiejiang-tvt}Y. Song and J. Xie, ``QB2IC: A QoS-Based Broadcast Protocol Under Blind Information for Multihop Cognitive Radio Ad Hoc Networks,'' \emph{IEEE Transactions on Vehicular Technology}, vol. 63, no. 3, pp. 1453--1466, Mar. 2014.
%
%\bibitem{xiejiang-tmc}Y. Song, J. Xie, and X. Wang, ``A Novel Unified Analytical Model for Broadcast Protocols in Multi-hop Cognitive Radio Ad Hoc Networks,'' \emph{IEEE Transactions on Mobile Computing}, vol. 13, no. 8, pp. 1653--1667, August 2014.
%
%\bibitem{xiejiang-bracer-tmc}Y. Song and J. Xie, ``BRACER: A Distributed Broadcast Protocol in Multi-hop Cognitive Radio Ad Hoc Networks with Collision Avoidance,'' \emph{IEEE Transactions on Mobile Computing}, 2014, to appear.

\bibitem{theis}N. Theis, R. Thomas, and L. DaSilva, ``Rendezvous
for cognitive radios,\textquotedblright{} \emph{IEEE Trans. Mobile
Computing}, vol. 10, no. 2, pp. 216--227, 2011.\end{thebibliography}
\end{document}